\documentclass[11pt,reqno]{amsart}

\usepackage{times}
\usepackage{scrextend}
\usepackage{amssymb, commath}
\usepackage{amsmath}
\usepackage{mathrsfs, mathtools, bbm}
\usepackage{enumerate}
\usepackage{amsthm, accents}
\usepackage{color}
\usepackage{hyperref}
\hypersetup{
    colorlinks=true, 
    linkcolor=blue,  
}

\newtheorem{theorem}{Theorem}[section]
\newtheorem{lemma}[theorem]{Lemma}
\newtheorem{proposition}[theorem]{Proposition}
\newtheorem{fact}[theorem]{Fact}
\newtheorem{problem}[theorem]{Problem}
\newtheorem{remark}[theorem]{Remark}
\theoremstyle{definition}
\newtheorem{definition}[theorem]{Definition}

\numberwithin{equation}{section}

\mathchardef\hyphen="2D

\newcommand{\ubar}[1]{\underaccent{\bar}{#1}}

\def \A {B}
\def \a {\alpha}

\def \e {\varepsilon}

\def \s {\sigma}

\def \E {\mathbb{E}}
\def \N {\mathbb{N}}

\def \R {\mathbb{R}}

\def \MM {\mathcal{M}}
\def \NN {\mathcal{N}}

\def \one {\mathbbm{1}}

\DeclareMathOperator{\diam}{diam}

\renewcommand{\Pr}[2][]{\mathbb{P}_{#1} \left\{ #2 \rule{0mm}{3mm}\right\}}

\title{Metric geometry of the privacy-utility tradeoff}

\author{March Boedihardjo}
\address{Department of Mathematics, Michigan State University, East Lansing, USA}
\email{boedihar@msu.edu}

\author{Thomas Strohmer}
\address{Department of Mathematics, University of California, Davis, USA}
\email{strohmer@math.ucdavis.edu}
\thanks{T.S. acknowledges support from NIH R01HL16351, NSF DMS-2027248,  and NSF DMS-2208356.}

\author{Roman Vershynin}
\address{Department of Mathematics, University of California, Irvine, USA}
\email{rvershyn@uci.edu}
\thanks{R.V. acknowledges support from NSF DMS-1954233, NSF DMS-2027299, U.S. Army 76649-CS, and NSF+Simons Research Collaborations on the
Mathematical and Scientific Foundations of Deep Learning.}

\begin{document}

\maketitle

\begin{abstract}
Synthetic data are an attractive concept to enable privacy in data sharing. A fundamental question is  how similar the privacy-preserving synthetic data are compared to the true data. Using metric privacy, an effective generalization of differential privacy beyond the discrete setting, we raise the problem of characterizing the optimal privacy-accuracy tradeoff by the metric geometry of the underlying space. We provide a partial solution to this problem in terms of the ``entropic scale'', a quantity that captures the multiscale geometry of a metric space via the behavior of its packing numbers. We illustrate the applicability of our privacy-accuracy tradeoff framework via a diverse set of examples of metric spaces.

\end{abstract}


\section{Introduction}

A compelling approach to enable privacy in data sharing is based on the concept
of synthetic data~\cite{bellovin2019privacy}. The goal of synthetic data is to create a dataset that maintains the statistical properties of the original data while not exposing sensitive information. A fundamental challenge is to derive optimal bounds on the achievable utility
of synthetic data while maintaining privacy.  Analyzing this privacy-utility tradeoff is the goal of this paper. 

There are numerous versions of this problem, depending on the specific notion of privacy and the choice of the utility metric.
A popular and rigorous framework to define and quantify privacy is differential privacy. While differential privacy is a concept of the discrete world (where datasets can differ in a single element), it is often necessary to have more flexibility in the type of input data. Metric privacy, introduced in~\cite{chatzikokolakis2013broadening}, provides this flexibility, as it generalizes differential privacy beyond the discrete setting. 
For that reason we will adopt metric privacy as our notion of privacy in this paper.
Metric privacy, tailored to the setting of synthetic data, is defined as follows:
\begin{definition}[Metric privacy]              \label{def: metric privacy}
  Let $(Z,\rho_1)$ be a  metric space and $\a>0$.
  A randomized algorithm $\MM: Z \to Z$
  is called $(\a,\rho_1)$-metrically private
  if, for any inputs $x,x' \in Z$ we have
  \begin{equation}      \label{eq: metric DP}
  \frac{\Pr{\MM(x) \in S}}{\Pr{\MM(x') \in S}} \le \exp \left( \a \, \rho_1(x,x') \right).
  \end{equation}
\end{definition}
Metric privacy has been utilized in various applications, such as in location privacy~\cite{ABCP} and privacy-preserving machine learning see e.g.~\cite{fernandes2019generalised,galli2023advancing}, as well as in theoretical studies of  synthetic data, cf.~\cite{boedihardjo2022measure}. Metric privacy includes classical differential privacy as special case, see e.g.~Lemma 4.1 in~\cite{boedihardjo2022measure}.

With our notion of privacy in place, how shall we quantify utility, i.e., measure accuracy? 
Naturally, we want to measure how much the privacy-preserving synthetic data $\MM(x)$ resemble the true data $x$. To that end we choose a (possibly different) metric $\rho_2$ and define:   

\begin{definition}[Optimal accuracy]
    Let $(Z,\rho_1,\rho_2)$ be a bimetric space\footnote{By a bimetric space we mean a set with two metrics defined on it.} and $\a>0$. Define the accuracy of the bimetric space as  
    $$
    A(Z,\alpha) = \inf_{\mathcal{M}} \sup_{x\in Z} \,\mathbb{E}\, \rho_{2}(\mathcal{M}(x),x),
    $$
    where the infimum is over all $(\alpha,\rho_1)$-metrically private mechanisms $\mathcal{M}: Z \to Z$.
\end{definition}

Thus, $A(Z,\alpha)$ gives the best accuracy of synthetic data for a given privacy budget. 

\begin{problem}   \label{problem: main}
    Express the accuracy $A(Z,\alpha)$ in terms of the geometry of the bimetric space $Z$.
\end{problem}

\subsection{State of the art}

The study of the privacy-accuracy tradeoff for classical differential privacy
has a long history, a detailed review of which is beyond the scope of this paper. 
The standard book on DP~\cite{dwork2014algorithmic} contains detailed discussions on achievable rates of accuracy for a given privacy budget.
The paper~\cite{ghosh2009universally} proves that for each fixed count query and a given differential
privacy budget, there exists a mechanism that is expected loss-minimizing.
Hardt and Talwar~\cite{hardttalwar} use methods from convex geometry to determine a nearly optimal tradeoff between privacy and accuracy, where the error is measured by the Euclidean distance between the correct answer and the actual answer of a linear query. 
Nikolov, Talwar and Zhang~\cite{nikolov2013geometry} extend these
results to the case of $(\epsilon,\delta)$-differential privacy.
In~\cite{gupte2010universally} Gupte and Sundararaja adopt a framework of risk-averse agents to derive a universal
optimality result  using a
minimax formulation.
Geng and Viswanath~\cite{geng2014optimal} give an  optimal differentially private mechanism for certain query functions under
a  general utility-maximization framework.
In~\cite{nikolov2023private} Nikolov derives differentially private mechanisms with optimal worst case sample complexity under mean square average error for statistical queries. 

Regarding the privacy-utilty tradeoff for  synthetic data, Ullman and Vadhan~\cite{ullman2011pcps} prove that (under standard cryptographic assumptions) there is no polynomial-time differentially private algorithm that can generate synthetic Boolean data such that all two-dimensional marginals  are close to those of the original Boolean dataset. In~\cite{boedihardjo2022covariance} the authors derive an computationally efficient algorithm to construct DP
 synthetic data with approximately optimal average error.
The paper~\cite{boedihardjo2022measure} establishes 
 asymptotically sharp privacy-accuracy bounds for synthetic data for general compact metric spaces.

The privacy-utility tradeoff for metric privacy is less well-understood than
for classical DP, since the achievable accuracy is strongly influenced by the geometry of the space.
In~\cite{bordenabe2014optimal} the tradeoff between location privacy and utility is analyzed by constructing for a given privacy budget a mechanism that minimizes utility loss  using linear programming techniques. 
In~\cite{imola2022balancing}, the authors consider linear programming-based metric privacy mechanisms that balance the tradeoff between utility and computational complexity. The paper~\cite{fernandes2022universal} investigates the privacy-accuracy tradeoff for metric privacy from an information theoretic perspective.

In this paper, we study the optimal privacy-accuracy tradeoff in terms of the geometry of the metric spaces, but we do not account for computational feasibility.

\subsection{Our contributions}

We do not know how to approach Problem~\ref{problem: main} in full generality. However, we can solve it in two special cases: 
when $\rho_1=\rho_2$ (Theorem~\ref{thm: same metric}) and when $\rho_2$ is an ultrametric (Theorem~\ref{thm: ultrametric}). 
In both results, the geometry of a bimetric space $(Z,\rho_1,\rho_2)$ is captured by the packing numbers\footnote{Working with packing numbers instead of covering numbers is simply a choice of convenience in this paper. Packing and covering numbers are equivalent (see e.g.~\cite[Lemma~4.2.8]{vershyninbook})~
so versions of all our results hold stated for covering numbers as well.} and diameters of the balls. 

To be specific, let $(Z,\rho)$ be a metric space. The {\em packing number} of $Z$  at scale $\e>0$ is denoted by $N(Z,\e)$. It is the maximal cardinality of an $\e$-separated subset $\NN \subset Z$, i.e. a subset in which all the points have distance strictly greater than $\e$ from each other. The {\em closed ball} centered at $x \in Z$ and with radius $r$ is denoted by $B(x,r)$. 

If $(Z,\rho_1,\rho_2)$ a bimetric space, we often use subscripts $1$ and $2$ to indicate which metric is being used. For example, $B_1(x,r)$ denotes the ball in the metric $\rho_1$, and $N_2(\A,\e)$ denotes the packing number of $\A$ in the metric $\rho_2$, where $\A\subset Z$.

Our crucial quantity is the scale on which the packing numbers of all balls are subexponential in radius: 

\begin{definition}[Entropic scale]  \label{def: entropic scale}
    The entropic scale of a bimetric space $(Z,\rho_1,\rho_2)$ is defined as follows:
    $$
    s(Z,\a)
    = \inf\left\{s>0 \;\Big\vert\; N_2 \left(B_1(x,r),s\right)\leq e^{\alpha r} \; \forall x\in Z, \; \forall r>0 \right\}.
    $$
\end{definition}

Our first result states that if $\rho_1=\rho_2=\rho$, the optimal accuracy is equivalent to the entropic scale: 

\begin{theorem}[Same metric]   \label{thm: same metric}
    Let $(Z,\rho)$ be a connected metric space and $\a>0$. Then\footnote{We use the notation $\lesssim$, $\gtrsim$ and $\asymp$ to hide positive absolute constant factors. Thus, $a \lesssim b$ means that $a \le Cb$, and $a \asymp b$ means that $ca \le b \le Cb$, where $C$ and $c$ are positive absolute constants. The specific numeric values of $C$ and $c$ can usually be derived from the proofs.} 
    $$
    s(Z,2\a)
    \lesssim A(Z,\a) 
    \lesssim s(Z,\a/3).
    $$
    Moreover, if $(Z,\rho)$ is norm-convex\footnote{A metric space $(Z,\rho)$ is called {\em norm-convex} if $Z$ is a convex subset of some normed space with the induced metric.} then 
    $$
    A(Z,\a) \asymp 
    s(Z,\a). 
    $$
\end{theorem}

The lower bound in Theorem~\ref{thm: same metric} holds in full generality, even when the metrics $\rho_{1}$ and $\rho_{2}$ are different, and it is proved in Theorem~\ref{thm: A>s}. The upper bound is proved in Theorem~\ref{thm: A<s}. The ``moreover'' part follows from regularity of the entropic scale, which is proved in Lemma~\ref{prop: regularity}.

If $(Z,\rho)$ is disconnected, then the entropic scale may fail to control accuracy. (Indeed, for a two-point metric space, all packing numbers are bounded by $2$, while the distance between the points and thus the accuracy may be arbitrarily large.) To fix this issue, we introduce the following cousin of the entropic scale:

\begin{definition}[Diametric scale]  \label{def: diametric scale}
    The diametric scale of a bimetric space $(Z,\rho_1,\rho_2)$ is defined as follows:
    $$
    s_\circ(Z,\a)
    \coloneqq \inf\left\{s>0 \;\Big\vert\; 
    \diam_2(B_1(x,r))\leq e^{\alpha r} s \; \forall x\in Z, \; \forall r>0 \right\},
    $$
where $\diam_2(B_1(x,r))$ is the diameter of $B_{1}(x,r)$ with respect to the metric $\rho_{2}$.
\end{definition}


Our second main result is concerned with ultrametrics. We note that ultrametrics have attracted increasing attention in data science in recent years, due to their usefullness in areas such as genomics and chemoinformatics~\cite{contreras2012fast,dragovich2017ultrametrics}. We 
 demonstrate that if $\rho_2$ is an ultrametric, the accuracy is equivalent to the sum of the entropic scale and the diametric scale:

\begin{theorem}[Ultrametric]\label{thm: ultrametric}
    Let $(Z,\rho_1,\rho_2)$ be a bimetric space. Suppose that $\rho_2$ is an ultrametric. Then for any $\a>0$ we have
    $$
    s(Z,2\a)+s_\circ(Z,2\a) 
    \lesssim A(Z,\a) 
    \lesssim s(Z,\a/3)+s_\circ(Z,\a/7).
    $$
\end{theorem}

The lower bound in Theorem~\ref{thm: ultrametric} holds in full generality, even if $\rho_2$ is a general metric, and it is proved in Theorems~\ref{thm: A>s} and \ref{thm: A>scirc}. The upper bound is proved in Theorem~\ref{thm: ultrametric full}.

\subsection{Examples}\label{s: examples intro}

Theorem~\ref{thm: same metric} and Theorem~\ref{thm: ultrametric}  allow us to compute the privacy-accuracy tradeoff for many metric spaces up to absolute constant factors. In particular:

\begin{enumerate}
    \item(Theorem~\ref{thm: unit ball}): If $Z$ is a closed {\em ball of a $d$-dimensional normed space}, then 
    \begin{equation}\label{eq: unit ball}
        A(Z,\a) \asymp s(Z,\a) \asymp \min \left( \frac{d}{\a}, 1\right).
    \end{equation}


   \item(Theorem~\ref{thm: Wasserstein}):
   If $Z$ is the set of all {\em probability measures} on the unit cube $[0,1]^d$, equipped with the $1$-Wasserstein metric with respect to the $\ell_\infty$-norm on $[0,1]^{d}$, then 
    \begin{equation}\label{eq: Wasserstein}
        A(Z,\a) \asymp s(Z,\a) \asymp \min\left( \a^{-\frac{1}{d+1}},1\right).
   \end{equation}

    \item(Theorem~\ref{thm: Lipschitz}): 
    If $Z$ is the set of all real-valued {\em $1$-Lipschitz functions} on the unit cube $[0,1]^d$ with $f(0)=0$, equipped with the $\ell^\infty$ metric, then \eqref{eq: Wasserstein} again holds. 

    \item(Remark~\ref{rem: Baire}):
    If $Z$ is the  {\em Boolean cube} $\{0,1\}^d$ equipped with the Baire ultrametric~\cite{contreras2012fast} (with base 2), then 
    \begin{equation}
    \label{ultraboolean}
        A(Z,\a) \asymp \min\left(\frac{1}{\a},1\right).
    \end{equation}
    A generalization of this example is obtained in Theorem \ref{thm:ultraexample}, where the two metrics $\rho_{1}$ and $\rho_{2}$ could be different.
    
\end{enumerate}

\subsection{Some open problems}\label{s:openproblems}

 The theorems  stated above motivate the following interesting  questions:
\begin{itemize}   
 \item[(i)] Can we solve Problem~\ref{problem: main} if we assume $\rho_2 \le \rho_1$?
\item[(ii)]  
    It seems plausible that the conclusion of Theorem~\ref{thm: ultrametric} holds if $\rho_1=\rho_2$ are the same metric but not necessarily an ultrametric. In other words,  can we drop the connectivity assumption in Theorem~\ref{thm: same metric} by including $s_\circ$ in the bound?
    \item[(iii)] 
Does the conclusion of Theorem~\ref{thm: ultrametric} hold for any bimetric space, perhaps up to logarithmic factors? 
   
    \end{itemize}

\subsection{Plan of the paper}

Section~\ref{s: basic properties} introduces some tools that can help to compute the entropic scale when the two metrics $\rho_{1}$ and $\rho_{2}$ are the same. For instance, in Subsection~\ref{s: doubling scale} we show that for a wide class of metric spaces (specifically, norm-convex metric spaces), one can always choose $r=2s$ in the Definition~\ref{def: entropic scale} of the entropic scale. In Subsection~\ref{s: outer scale}, we notice another simplification of the definition of entropic scale: it is often enough to compute the packing numbers of the entire metric space $Z$ rather than of all balls. 
In Section~\ref{s: lower bounds}, we show that accuracy is always bounded below by the entropic scale (Theorem~\ref{thm: A>s}) and the diametric scale (Theorem~\ref{thm: A>scirc}). This establishes the lower bounds in both of our main results Theorems~\ref{thm: same metric} and \ref{thm: ultrametric}. 
In Section~\ref{s: upper bound: same metric}, we show that if $\rho_1=\rho_2$, the accuracy is bounded above by the entropic scale, thus completing the proof of Theorem~\ref{thm: same metric}.
In Section~\ref{s: ultrametric}, we show that if $\rho_2$ is an ultrametric, the accuracy is bounded above by the sum of entropic and diametric scales, thus completing the proof of Theorem~\ref{thm: ultrametric}.
In Section~\ref{s: examples}, we work out the examples announced in Section~\ref{s: examples intro}.

\section{Basic properties of the entropic scale: same metric}\label{s: basic properties}
Throughout this section, we assume that the two metrics $\rho_{1}=\rho_{2}=\rho$ are the same.
\subsection{Background on packing and covering}

\begin{fact}[Packing implies covering]  \label{fact: covering}
    Let $(Z,\rho)$ be a metric space and $\A \subset Z$ be any subset. Then any maximal $\e$-separated subset $\NN \subset \A$ is an $\e$-cover of $\A$, i.e. 
    $$
    \forall x \in \A \; \exists y \in \NN: \; \rho(x,y) \le \e.
    $$
\end{fact}

\begin{proof}
    For contradiction, assume that there exists $x \in \A$ such that $\rho(x,y)>\e$ for all $y \in \NN$. Thus $x \not\in\NN$, and $\NN \cup \{x\}$ is an $\e$-separated subset of $\A$. This contradicts the maximality of $\NN$.
\end{proof}

\begin{fact}[Chain rule] \label{fact: chain rule one term}
    Let $(Z,\rho)$ be a metric space and $\A \subset Z$ be any subset. Then for any $r,s>0$ we have
    $$
    N(\A,s) 
    \le N(\A,r) \cdot \sup_{a \in \A} N\left(B(a,r),s\right).
    $$
\end{fact}

\begin{proof}
    Let $\NN_s \subset \A$ be any $s$-separated subset. Let $\NN_r \subset \A$ be a maximal $r$-separated subset. 
    Then $\A$ can be covered by $\abs{\NN_r} \le N(\A,r)$ balls of radius $r$ centered at points in $\A$. Each such ball $B(a,r)$ contains at most $N\left(B(a,r),s\right)$ points from $\NN_s$. Therefore
    $$
    \abs{\NN_s} \le N(\A,r) \cdot \sup_{a \in \A} N\left(B(a,r),s\right).
    $$
    The proof is complete. 
\end{proof}

The chain rule yields the following for any $s>0$:
$$
\sup_{x \in Z} N\left(B(x,4s),s\right)
\le \sup_{x \in Z} N\left(B(x,4s),2s\right) \cdot \sup_{y \in Z} N\left(B(y,2s),s\right).
$$
Iterating this bound, we obtain:

\begin{fact}[Chain rule with many terms]\label{fact: chaining}
    Let $(Z,\rho)$ be a metric space. Then for any $s>0$ and $k_0 \in \N$ we have
    $$
    \sup_{x \in Z} N\left(B(x,s2^{k_0}),s\right)
    \le \prod_{k=1}^{k_0} \sup_{x_k \in Z} N\left(B(x_k,s2^k),s2^{k-1}\right).
    $$
\end{fact}




\subsection{Regularity of entropic scale}

\begin{proposition}[Regularity]   \label{prop: regularity}
    Let $(Z,\rho)$ be a norm-convex metric space and $\a>0$. Then
    $$
    s(Z,\a)
    \le s(Z,t\a) 
    \le\frac{1}{t} s(Z,\a)
    \quad \text{for all } t\in[0,1].
    $$
\end{proposition}

The proof is based on the following monotonicity property:

\begin{lemma}[Monotonicity]   \label{lem: monotonicity}
    Let $(Z,\rho)$ be a norm-convex metric space. Then for all $x\in Z$, $t\in[0,1]$ and $r,s>0$, we have
    $$
    N\left(B(x,r),s\right)
    \le N\left(B(x,tr),ts\right).
    $$
\end{lemma}

\begin{proof}
Norm-convexity implies that 
$B(x,tr) \supset (1-t)x+tB(x,r)$.
Therefore
$$
N\left(B(x,tr),ts\right)
\ge N\left((1-t)x+tB(x,r),ts\right)
=N\left(B(x,r),s\right). \qedhere
$$
\end{proof}

\begin{proof}[Proof of Proposition~\ref{prop: regularity}]
The lower bound is trivial by the definition of entropic scale. To check the upper bound, set $s \coloneqq s(Z,\a)$ and fix any $x \in Z$. The definition of entropic scale guarantees that 
$$
N\left(B(x,r),s\right) \le e^{\a r}
\quad \text{for all } r>0.
$$
Using Lemma~\ref{lem: monotonicity} for $s/t$ instead of $s$, and then the bound above for $tr$ instead of $r$, we obtain
$$
N\left(B(x,r),s/t\right)
\le N\left(B(x,tr),s\right)
\le e^{\a \cdot tr}\quad \text{for all } r>0.
$$
By the definition of entropic scale, this yields $s(Z,t\a) \le s/t$. The proof is complete.
\end{proof}

\subsection{Elementary lower bounds for entropic scale}

\begin{lemma}   \label{lem: 1/a large}
    Let $(Z,\rho)$ be a metric space and $\a>0$. Then 
    $$
    \frac{1}{\a} \ge 2\diam(Z) 
    \quad \text{implies} \quad
    s(Z,\a) \ge \frac{1}{2} \diam(Z).
    $$
\end{lemma}

\begin{proof}
    Assume that the conclusion is false. By the definition of entropic scale, this means that 
    $$
    N \left( B(x,r), \frac{1}{2}\diam(Z) \right)  \le e^{\a r}
    \quad \text{for any $x \in Z$ and $r>0$}.
    $$
    Let us use this bound for $r = \diam(Z)$. We trivially have $B(x,r)=Z$, and $\a r \le 1/2$ holds by assumption. It follows that
    $$
    N \left( Z, r/2 \right) 
    \le e^{1/2}
    < 2. 
    $$
    This means that the distance between any pair of points in $Z$ must be bounded by $r/2$. (Otherwise a pair of $r/2$-separated points $x,y$ would make the packing number at least $2$.) But this shows that 
    $\diam(Z) \le r/2 = \diam(Z)/2$, a contradiction.
\end{proof}

\begin{lemma}    \label{lem: 1/a small}
    Let $(Z,\rho)$ be a connected metric space and $\a>0$. Then 
    $$
    \frac{1}{\a} < 2\diam(Z) 
    \quad \text{implies} \quad
    s(Z,\a) \ge \frac{1}{2\a}.
    $$
\end{lemma}

\begin{proof}
    By assumption, there exists a pair of points $x,y \in Z$ satisfying 
    \begin{equation}    \label{eq: x y separated}
        \rho(x,y) > \frac{1+\e}{2\a}
        \quad \text{for some } \e>0.
    \end{equation}
    Assume that the conclusion of the lemma is false. By definition of entropic scale, this means that 
    $$
    N \left( B(x,r), \frac{1}{2\a} \right)  \le e^{\a r}
    \quad \text{for any } r>0.
    $$
    Plug $r = 2/(3\a)$ to get
    $$
    N \left( B \Big(x,\frac{2}{3\a}\Big), \frac{1}{2\a} \right) 
    \le e^{2/3}
    < 2. 
    $$
    This means that every point $y$ in the ball $B(x,2/(3\a))$ must satisfy $\rho(x,y) \le 1/(2\a)$. (Otherwise a pair of $1/(2\a)$-separated points $x,y$ would make the packing number at least $2$.) Thus we showed that 
    $$
    B \Big(x,\frac{2}{3\a}\Big)
    = B \Big(x,\frac{1}{2\a}\Big).
    $$
    This means that any ball of the intermediate radius, and in particular the ball $B \coloneqq B(x,\frac{1+\e}{2\a})$, must be closed and open. Moreover, \eqref{eq: x y separated} shows that $y \not\in B$, so $B$ is a proper subset of $Z$. This contradicts the connectedness of $Z$. 
\end{proof}

\subsection{Doubling scale}\label{s: doubling scale}

Fixing $r=2s$ in the definition of entropic scale (Definition~\ref{def: entropic scale}) leads to:

\begin{definition}[Doubling scale]
    The {\em doubling scale} of a metric space $(Z,\rho)$ is defined as follows:
    $$
    \ubar{s}(Z,\a)
    = \inf\left\{s>0 \;\Big\vert\; N\left(B(x,2s),s\right)\leq e^{2\alpha s} \; \forall x\in Z \right\}.
    $$
\end{definition}

\begin{proposition}[Entopic vs doubling scale] \label{prop: doubling}
    For any norm-convex metric space $(Z,\rho)$ and any $\a>0$, we have
    $$
    \frac{1}{4} s(Z,\a) 
    \le \ubar{s}(Z,\a) 
    \le s(Z,\a).
    $$
\end{proposition}

\begin{proof}
    The upper bound is trivial. 
    To prove the lower bound, let $s=\ubar{s}(Z,\a)$. Fix any $x \in Z$, and let us bound the packing number
    $$
    N(r) \coloneqq N\left(B(x,r),s\right)
    $$
    for each $r>0$. Consider three cases. 

    1. If $r \in (0,s/2)$, then $N(r)=1$.

    2. If $r \in [s/2,2s)$, then the definition of $s$ gives
    $$
    N(r) \le N(B(x,2s),s) \le e^{2\a s} \le e^{4\a r}.
    $$

    3. If $r \ge 2s$, then $r \in [s2^{k_0-1},s2^{k_0})$ for some $k_0 \in \{2,3,\ldots\}$. Then by the chain rule (Fact~\ref{fact: chaining}) we have 
    $$
    N(r) \le \prod_{k=1}^{k_0} \sup_{x_k \in Z} N\left(B(x_k,s2^k),s2^{k-1}\right).
    $$
    Using Lemma~\ref{lem: monotonicity} and the definition of $s$, we see that each factor in this product is bounded by $N(B(x,2s),s) \le e^{2\a s}$. Thus
    $$
    N(r) \le e^{2\a s k_0} \le e^{2\a r},
    $$
    where in the last step we used that $k_0 \le 2^{k_0-1}$.

    Combining all three cases, we conclude that 
    $N(r) \le e^{4\a r}$ for any $r>0$. By the definition of entropic scale, this implies $s(Z,4\a) \le s$. By Proposition~\ref{prop: regularity}, we get
    $$
    s(Z,\a) \le 4s(Z,4\a) \le 4s. 
    $$
    In view of the definition of $s$, this completes the proof. 
\end{proof}

\subsection{Outer scale}\label{s: outer scale}

\begin{definition}[Outer scale]
    The {\em outer scale} of a metric space $(Z,\rho)$ is defined as follows:
    $$
    \bar{s}(Z,\a)
    = \inf_{\gamma>0} \left\{\gamma + \frac{1}{\a} \ln N(Z,\gamma)\right\}.
    $$
\end{definition}

\begin{proposition}\label{prop: s<sbar}
    For any metric space $(Z,\rho)$ and any $\a>0$, we have
    $$
    s(Z,\a) \le 2\bar{s}(Z,\a).
    $$
\end{proposition}

\begin{proof}
    Fix any $\gamma>0$ and let $s \coloneqq \gamma + \frac{1}{\a}\ln N(Z,\gamma)$. Fix any $x \in Z$, and let us bound the packing number
    $$
    N(r) \coloneqq N(B(x,r),2s)
    $$
    for each $r>0$. Consider two cases. 

    1. If $r<s$, then $N(r)=1$.

    2. If $r \ge s$, then, using that $2s \ge \gamma$ and by definition of $s$ we have  
    $$
    N(r) \le N(Z,\gamma) \le e^{\a s} \le e^{\a r}.
    $$

    Thus $N(r) \le e^{\a r}$ for any $r>0$. By the definition of entropic scale, this implies $s(Z,\a) \le 2s$. In view of the definition of $s$, this completes the proof.
\end{proof}

For norm-convex spaces, the bound in Proposition~\ref{prop: s<sbar} can be reversed up to a logarithmic factor:

\begin{proposition}\label{prop: sbar<slog}
    For any norm-convex metric space $(Z,\rho)$ and any $\a>0$, we have
    $$
    \bar{s}(Z,\a) \le 3s(Z,\alpha) \left[1 + \ln \frac{\diam(Z)}{s(Z,\a)}\right].
    $$
\end{proposition}

\begin{proof}
    Let $s \coloneqq s(Z,\a)$ and 
    $$
    k_0 \coloneqq \left\lceil \log_2 \frac{\diam(Z)}{s}\right\rceil.
    $$
    Then $\diam(Z) \le s2^{k_0}$ and thus $Z \subset B(x,s2^{k_0})$ for any $x \in Z$. Then by chain rule (Fact~\ref{fact: chaining}) we have 
    $$
    N(Z,s) \le \prod_{k=1}^{k_0} \sup_{x_k \in Z} N\left(B(x_k,s2^k),s2^{k-1}\right).
    $$
    Using Lemma~\ref{lem: monotonicity} and the definition of $s$, we see that each factor in this product is bounded by $N(B(x,2s),s) \le e^{2\a s}$. Thus
    $$
    N(Z,s) \le e^{2\a s k_0}.
    $$
    It follows that 
    $$
    s+\frac{1}{\a}\ln N(Z,s) 
    \le s+2sk_0
    \le 3s \left[1 + \ln \frac{\diam(Z)}{s} \right].
    $$
    Recall the definitions of $s$ and $\bar{s}(Z,\a)$ to complete the proof.
\end{proof}

\begin{remark}
    The logarithmic factor in Proposition~\ref{prop: sbar<slog} cannot be removed in general. For example, let $Z$ be the unit Euclidean ball in $\R^d$ with the Euclidean metric, and let $\a>2d$. Then, as we will see in Theorem~\ref{thm: unit ball}, $s(Z,\a) \asymp d/\a$ while it is not hard to check that $\bar{s}(Z,\a) \asymp (d/\a) \log(\a/d)$.
\end{remark}

Nevertheless, the logarithmic factor can be removed if the packing numbers satisfy a  {\em doubling condition}:

\begin{proposition}\label{prop: sbar<s}
    Let $(Z,\rho)$ be a norm-convex metric space and let $\a>0$. Set $ s=s(Z,\a)$. Assume that 
    $$
    N(Z,s) \ge N(Z,\kappa s)^2
    \quad \text{for some } \kappa \ge 1.
    $$
    Then 
    $$
    \frac{s}{2} \le \bar{s}(Z,\a) \le 2 \kappa s.
    $$
\end{proposition}

For the proof, we need: 

\begin{lemma}\label{lem: pack by sbar}
    Let $(Z,\rho)$ be a norm-convex metric space and let $\a>0$. Set $ s=s(Z,\a)$. Then
    $$
    \frac{1}{\a} \ln N(Z,s) \le \bar{s}(Z,\a).
    $$
\end{lemma}

\begin{proof}
    Fix any $\gamma>0$ and use the chain rule (Fact~\ref{fact: chain rule one term})  and the definition of $s$ to get
    $$
    N(Z,s) 
    \le N(Z,\gamma) \cdot \sup_{z \in Z} N\left( B(z,\gamma),s\right)
    \le N(Z,\gamma) \cdot e^{\a\gamma}.
    $$
    Rearranging the terms yields
    $$
    \frac{1}{\a} \ln N(Z,s) 
    \le \gamma + \frac{1}{\a} \ln N(Z,\gamma).
    $$
    Take the infimum over $\gamma>0$ to complete the proof. 
\end{proof}

\begin{proof}[Proof of Proposition~\ref{prop: sbar<s}]
    The lower bound was proved in Proposition~\ref{prop: s<sbar}, so it is enough to prove the upper bound. Using first the doubling condition, then Lemma~\ref{lem: pack by sbar}, and finally the definition of the outer scale, we get
    $$
    \frac{2}{\a} \ln N(Z,\kappa s)
    \le \frac{1}{\a} \ln N(Z,s)
    \le \bar{s}(Z,\a)
    \le \kappa s + \frac{1}{\a} \ln N(Z,\kappa s).
    $$
    Rearranging the terms, we get
    $$
    \kappa s + \frac{1}{\a} \ln N(Z,\kappa s) \le 2 \kappa s.
    $$
    By the definition of the outer scale, this yields $\bar{s}(Z,\a) \le 2\kappa s$.
    The proposition is proved.
\end{proof}

\section{Lower bounds on accuracy}\label{s: lower bounds}

\begin{theorem} \label{thm: A>s}
    For any bimetric space $(Z,\rho_1,\rho_2)$ and any $\alpha>0$, we have
    $$
    A(Z,\alpha) \ge \frac{1}{8} s(Z,2\alpha).
    $$
\end{theorem}

\begin{proof}
    Fix a point $x \in Z$ and numbers $r,s>0$.
    By the definition of packing numbers, there exists a subset $S(x,r) \subset B_1(x,r)$ that is $(8s)$-separated in the metric $\rho_2$ and has cardinality
    $$
    \abs{S(x,r)} = N_2\left( B_1(x,r),8s \right).
    $$
    The separation condition implies by triangle inequality that all balls $B_2(y,4s)$ centered at points $y \in S(x,r)$ are disjoint. Let $\MM:Z\to Z$ be a $(\alpha,\rho_1)$-metrically private randomized algorithm. We have
    $$
    1 = \Pr{\MM(x) \in Z}
    \ge \sum_{y \in S(x,r)} \Pr{\MM(x) \in B_2(y,4s)}.
    $$
    Since the mechanism $\MM$ is $(\alpha,\rho_1)$-metrically private, we have
    $$
    \Pr{ \MM(x) \in B_2(y,4s)}
    \ge e^{-\a \rho_1(x,y)} \cdot \Pr{\MM(y) \in B_2(y,4s)}.
    $$
    Now, for $y \in S(x,r) \subset B_1(x,r)$, we have 
    $e^{-\a \rho_1(x,y)} \ge e^{-\a r}$. Furthermore, Markov's inequality yields
    $$
    \Pr{\MM(y) \in B_2(y,4s)} 
    = \Pr{\rho_2(\MM(y),y) \le 4s}
    \ge \frac{3}{4}
    $$
    whenever $s \ge \E \rho_2(\MM(y),y)$. 
    Combining these bounds, we conclude the following. If 
    $$
    s = \sup_{y \in Z} \E \rho_2(\MM(y),y),
    $$
    then 
    $$
    1 \ge \abs{S(x,r)} \cdot e^{-\a r} \cdot \frac{3}{4}.
    $$
    Rearranging the terms, we get
    $$
    N_2\left( B_1(x,r),8s \right) 
    = \abs{S(x,r)}
    \le \left\lfloor \frac{4}{3} e^{\a r} \right\rfloor \le e^{2\a r}.
    $$
    By the definition of entropic scale, this yields
    $$
    s(Z,2\a) \le 8s = 8\sup_{y \in Z} \E \rho_2(\MM(y),y).
    $$
    Take the infimum over all $(\alpha,\rho_1)$-metrically private mechanisms $\MM$ on both sides to complete the proof.
\end{proof}

\begin{theorem}\label{thm: A>scirc}
    For any bimetric space  $(Z,\rho_1,\rho_2)$ and any $\a>0$, we have
    $$
    A(Z,\a) \ge \frac{1}{5} s_\circ(Z,2\a).
    $$
\end{theorem}

\begin{proof}
    Fix any $(\a,\rho_{1})$-metrically private mechanism $\MM$ and any $r>0$. Markov's inequality gives
    $$
    \Pr{\rho_2(\MM(x),x) > 2e^{\a r}s} 
    \le \frac{1}{2e^{\a r}}
    \quad \text{whenever } 
    s \ge \E \rho_2(\MM(x),x).
    $$
    Since $\MM$ is $(\a,\rho_{1})$-metrically private, the probability bound above yields
    $$
    \Pr{\rho_2(\MM(y),x) > 2e^{\a r}s} 
    \le \frac{1}{2}
    \quad \text{where } 
    r=\rho_1(x,y).
    $$
    On the other hand, Markov's inequality gives
    $$
    \Pr{\rho_2(\MM(y),y) > 3s} 
    \le \frac{1}{3}
    \quad \text{whenever } 
    s \ge \E \rho_2(\MM(y),y).
    $$
    Combining the last two probability bounds by triangle inequality, we obtain 
    $$
    \Pr{\rho_2(x,y) > 2e^{\a r}s+3s} \le \frac{1}{2}+\frac{1}{3} < 1.
    $$
    But for fixed $x,y$, the event in the left hand side is deterministic. Thus, the last bound implies that 
    $$
    \rho_2(x,y) 
    \le (2e^{\a r}+3)s \le 5e^{\a r}s.
    $$
    Summarizing, we have shown that
    $$
    \rho_2(x,y) \le 5e^{\a \rho_1(x,y)}s
    \quad \text{ where } s=\sup_{x \in Z} \E \rho_2(\MM(x),x).
    $$
    
    Now consider a ball $B_1(x,r)$ with an arbitrary center $x \in Z$ and arbitrary radius $r>0$. Since the $\rho_1$-distance between any pair of points in this ball is at most $2r$, it follows that 
    $$
    \diam_2(B_1(x,r)) \le 5e^{2\a r}s.
    $$
    Since this holds for all $x \in Z$ and $r>0$, we proved that 
    $$
    s_\circ(Z,2\a) 
    \le 5s 
    = 5 \sup_{x \in Z} \E \rho_2(\MM(x),x).
    $$
    Since this holds for any $(\a,\rho_{1})$-metrically private mechanism $\MM$, the proof is complete. 
\end{proof}

\section{Upper bounds on accuracy: same metric}\label{s: upper bound: same metric}
Throughout this section, we assume that the two metrics $\rho_{1}=\rho_{2}=\rho$ are the same.

\begin{lemma}[Exponential mechanism]    \label{lem: exponential mechanism}
    For any metric space $(Z,\rho)$ and $\a>0$, we have
    $$
    A(Z,\a) 
    \lesssim s(Z,\a/3) + \frac{1}{\a}.
    $$
\end{lemma}

\begin{proof}
    Let 
    \begin{equation}    \label{eq: s alpha/3}
        s \coloneqq s(Z,\a/3).   
    \end{equation}
    Let $\NN$ be a maximal $s$-separated subset of $Z$. Define a randomized map $\MM: Z \to \NN$ as follows. Given an input $x \in Z$, output $y \in \NN$ with probability proportional to $e^{-\a\rho(x,y)/2}$. In other words, for each $x \in Z$, we let 
    $$
    \Pr{\MM(x)=y}
    \coloneqq \frac{1}{\Sigma(x)} e^{-\a\rho(x,y)/2} 
    \quad \text{for all } y \in \NN,
    $$
    where 
    \begin{equation}    \label{eq: Sigma}
        \Sigma(x) = \sum_{y \in \NN} e^{-\a\rho(x,y)/2}.
    \end{equation}
    It is clear that this mechanism $\mathcal{M}$ is $(\alpha,\rho)$-metrically private (this follows analogously to establishing differential privacy of  the usual exponential mechanism, see~\cite[Chapter 3.4]{dwork2014algorithmic}).

    To estimate the accuracy of $\MM$, we need to fix any input $x \in Z$ and bound the expected value of the random variable
    $$
    A \coloneqq \rho(x,\MM(x)).
    $$
    We start by the trivial bound
    \begin{equation}    \label{eq: EA decomposed}
        \E A \le 6s + \E A \one_{\{A>6s\}} 
    \end{equation}
    By definition of $\MM(x)$, we have
    $$
    \E A \one_{\{A>6s\}}
    = \frac{1}{\Sigma(x)} \sum_{y \in \NN:\; \rho(x,y)>6s} \rho(x,y) \, e^{-\a\rho(x,y)/2}.
    $$
    Let us decompose the set of vectors $y$ in this sum according to their distance from $x$ as follows:
    $$
    \left\{y \in \NN:\; \rho(x,y)>6s\right\}
    = \bigsqcup_{k=1}^\infty \NN_k
    $$
    where
    $$
    \NN_k = \left\{y \in \NN:\; 6s+\frac{k-1}{\a} < \rho(x,y) \le 6s+\frac{k}{\a}\right\}.
    $$
    Thus
    $$
    \E A \one_{\{A>6s\}}
    = \frac{1}{\Sigma(x)} \sum_{k=1}^\infty \sum_{y \in \NN_k} \left(6s+\frac{k}{\a}\right) e^{-\a\left(6s+(k-1)/\a\right)/2}.
    $$
    
    Recall that by construction, $\NN_k$ is an $s$-separated subset of the ball $B(x,6s+k/\a)$. Thus, due to our choice of $s$ in \eqref{eq: s alpha/3} and by the definition of entropic scale, we have
    \begin{equation}    \label{eq: Nk}
        \abs{\NN_k} \le e^{(\a/3)(6s+k/\a)}
        = e^{2\a s + k/3}. 
    \end{equation}
    Next, let us find a lower bound on $\Sigma$, a quantity we defined in \eqref{eq: Sigma}. By definition of $\NN$ and Fact~\ref{fact: covering}, there exists $y \in \NN$ that satisfies $\rho(x,y) \le s$. It follows that 
    \begin{equation}    \label{eq: Sigma lower}
        \Sigma(x) \ge e^{-\a \rho(x,y)/2} \ge e^{-\a s/2}.
    \end{equation}
    Then using \eqref{eq: Nk} and \eqref{eq: Sigma lower}, we conclude that
    \begin{align*}
        \E A \one_{\{A>6s\}}
        &\le e^{\a s/2} \sum_{k=1}^\infty e^{2\a s + k/3} \left(6s+\frac{k}{\a}\right) e^{-3\a s-(k-1)/2} \\
        &= e^{(1-\a s)/2} \left( 6s\sum_{k=1}^\infty e^{-k/6} + \frac{1}{\a} \sum_{k=1}^\infty k e^{-k/6} \right)
        \lesssim e^{-\a s/2} \left(s+\frac{1}{\a}\right).
    \end{align*}
    Plugging this bound into \eqref{eq: EA decomposed}, we arrive at
    $$
    \E A \lesssim s + e^{-\a s/2} \left(s+\frac{1}{\a}\right).
    $$
    
    Summarizing, we proved that 
    $$
    A(Z,\a) \lesssim s + e^{-\a s/2} \left(s+\frac{1}{\a}\right)
    \quad \text{where } s=s(Z,\a/3).
    $$
    Since $\a s \ge 0$, this bound is stronger that than the one announced in the statement of the lemma.
\end{proof}

\begin{theorem} \label{thm: A<s}
    Let $(Z,\rho)$ be a connected metric space and $\a>0$. Then 
    $$
    A(Z,\alpha) \lesssim s(Z,\alpha/3).
    $$
\end{theorem}

\begin{proof}
    Assume first that $3/\alpha \ge 2\diam(Z)$. In this case, we can use a trivial mechanism $\MM$ that always outputs the same, arbitrarily chosen point in $Z$. This gives accuracy 
    $A(Z,\alpha) \le \diam(Z,\rho)$, while 
    Lemma~\ref{lem: 1/a large} yields $s(Z,\a/3) \ge \frac{1}{2}\diam(Z)$. The desired bound follows.

    Next, assume that $3/\alpha < 2\diam(Z)$. In this case, 
    we can use the exponential mechanism, for which Lemma~\ref{lem: exponential mechanism} gives accuracy $A(Z,\a) 
    \lesssim s(Z,\a/3) + 1/\a$. On the other hand, Lemma~\ref{lem: 1/a small} yields $s(Z,\a/3) \ge 3/(2\a)$.
    Combining these two bounds gives the conclusion. 
\end{proof}

\section{Upper bound on accuracy: ultrametric}\label{s: ultrametric}

\begin{theorem}[Ultrametric]\label{thm: ultrametric full}
    let $(Z,\rho_1,\rho_2)$ be a bimetric space such that $\rho_2$ is an ultrametric. Then for any $\a>0$ we have
    $$
    A(Z,\a) \lesssim s(Z,\a/3)+s_\circ(Z,\a/7).
    $$
\end{theorem}

The proof of this result will generally follow the proof of Theorem~\ref{lem: exponential mechanism}, but the exponential mechanism will be ``relaxed''. 

As before, let
\begin{equation}\label{eq: s ultra}
    s \coloneqq s(Z,\a/3).
\end{equation}

Let $\NN$ be a maximal $s$-separated subset of $Z$ in the $\rho_2$ ultrametric.

For any two points $x,y \in Z$, let $\s(x,y)$ denote the $\rho_1$-distance from $x$ to the $\rho_2$-ball centered at $y$ and with radius $s$, that is 
$$
\s(x,y) \coloneqq \inf\left\{ \rho_1(x,v):\; v \in B_2(y,s) \right\}.
$$
Define a randomized map $\MM: Z \to \NN$ as follows. Given an input $x \in Z$, output $y \in \NN$ with probability proportional to $e^{-\a\s(x,y)/2}$. In other words, for each $x \in Z$, we let 
\begin{equation}\label{eq: M ultra}
    \Pr{\MM(x)=y}
    \coloneqq \frac{1}{\Sigma(x)} e^{-\a\s(x,y)/2} 
    \quad \text{for all } y \in \NN,
\end{equation}
where 
\begin{equation*}    
    \Sigma(x) = \sum_{y \in \NN} e^{-\a\s(x,y)/2}.
\end{equation*}
Since $\s(x',y) \le \s(x,y)+\rho_1(x,x')$ for all $x,x',y \in Z$, the mechanism $\mathcal{M}$ is $(\alpha,\rho_{1})$-metrically private (this follows analogously to proving differential privacy of  the usual exponential mechanism, see~\cite[Chapter 3.4]{dwork2014algorithmic}).

By the definition of $\NN$ and Fact~\ref{fact: covering}, for each $x \in Z$ there exists $y \in \NN$ that satisfies $\rho_2(x,y) \le s$. It follows that $\sigma(x,y)=0$ and thus
\begin{equation}\label{eq: Sigma large}
    \Sigma(x) \ge 1. 
\end{equation}

We will first estimate the accuracy of $\MM$ in the ``relaxed distance'' $\s$, and then remove relaxation and transfer the result to the original ultrametric $\rho_2$. The following lemma is crucial; it is the only step of the argument where the ultrametric is used.

\begin{lemma}[Relaxed ball]\label{lem: relaxed ball}
    For any $x \in Z$ and $r>0$, we have
    $$
    \abs{\left\{ y \in \NN:\, \s(x,y)<r \right\}}
    \le N_2\left(B_1(x,r),s\right)
    \le e^{\a r/3}.
    $$
\end{lemma}

\begin{proof}
    Consider the set 
    $$
    \NN_r \coloneqq \left\{ y \in \NN:\, \s(x,y)<r \right\}.
    $$
    By the definition of $\sigma$, for each $y \in \NN_r$ there exists $v=v(y) \in Z$ that satisfies 
    \begin{equation}    \label{eq: v(y)} 
        \rho_1\left(x,v(y)\right)<r \quad \text{and} \quad \rho_2\left(v(y),y\right) \le s.
    \end{equation}
    Moreover, the definition of $\NN$ implies that $\rho_2(y,y')>s$ for any pair of distinct points $y,y' \in \NN_r\subset\NN$. We can combine this lower bound with the two upper bounds $\rho_2\left(v(y),y\right) \le s$ and $\rho_2\left(v(y'),y'\right) \le s$ using the max-triangle inequality for the ultrametic $\rho_2$. This gives 
    \begin{equation}    \label{eq: s-separated}
        \rho_2(\left(v(y),v(y')\right) > s.
    \end{equation}
    One consequence of this bound is that the map $y \mapsto v(y)$ is injective, so 
    \begin{equation}    \label{eq: Nr<Mr}
        \abs{\NN_r} \le \abs{v(\NN_r)}. 
    \end{equation}
    Moreover, the set $v(\NN_r)$ is $s$-separated in ultrametric $\rho_2$ by \eqref{eq: s-separated} and is contained in the ball $B_1(x,r)$ by the first inequality in \eqref{eq: v(y)}. Thus, 
    $$
    \abs{v(\NN_r)} \le N_2\left(B_1(x,r),s\right).
    $$
    Combining this with \eqref{eq: Nr<Mr}, we obtain the first bound in the statement of the lemma. 
    The second bound follows from the definition of $s$ in \eqref{eq: s ultra}.
\end{proof}

Next we analyze the accuracy of the mechanism $\MM$ in the relaxed distance $\s$. We shall prove that
$$
\E\, \s(x,\MM(x)) \lesssim \frac{1}{\a}
\quad \text{for any } x \in Z.
$$
In fact we can prove a stronger tail bound, from which the bound on expectation follows immediately: 

\begin{lemma}[Relaxed accuracy]\label{lem: relaxed accuracy}
    For any $x \in Z$ and $r>0$, we have
    $$
    \Pr{ \s(x,\MM(x)) \ge r} \lesssim e^{-\a r/6}.
    $$
\end{lemma}

\begin{proof}
    We have
    $$
    \Pr{\s(x,\MM(x)) \in [r,r+1/\a)}
    = \sum_{y \in \NN:\, \s(x,y) \in [r,r+1/\a)} \Pr{\MM(x)=y}.
    $$
    By the definition of mechanism $\MM$ in \eqref{eq: M ultra} and since $\Sigma(x) \ge 1$ by \eqref{eq: Sigma large}, each term of the sum above is bounded by $e^{-\a\s(x,y)/2} \le e^{-\a r/2}$. By Lemma~\ref{lem: relaxed ball}, the number of terms in the sum is bounded by 
    $$
    \abs{\left\{y \in \NN:\, \s(x,y)<r+1/a\right\}}
    \le e^{\a(r+1/\a)/3} \le 2e^{\a r/3}.
    $$
    Thus, we have proved that
    $$
    \Pr{\s(x,\MM(x)) \in [r,r+1/\a)} 
    \le e^{-\a r/2} \cdot 2e^{\a r/3}
    = 2e^{-\a r/6}
    $$
    for any $r>0$.
    Using the above bound for $r+k/\a$ instead of $r$, we conclude that
    \begin{align*}
        \Pr{\s(x,\MM(x)) \ge r}
        &= \sum_{k=0}^\infty \Pr{\s(x,\MM(x)) \in \left[r+\frac{k}{\a},r+\frac{k+1}{\a}\right)}\\
        &\le \sum_{k=0}^\infty 2e^{-\a(r+k/\a)/6}
        = 2e^{-\a r/6} \sum_{k=0}^\infty e^{-k/6}
        \lesssim e^{-\a r/6}.
    \end{align*}
\end{proof}

The next lemma allows one to transfer accuracy bounds from the relaxed distance $\sigma$ to the original ultrametric $\rho_2$. 

\begin{lemma}[Unrelaxation]\label{lem: unrelaxation}
    Let $x,y \in Z$ and $r>0$. Then
    $$
    \s(x,y)<r \quad \text{implies} \quad \rho_2(x,y) \le s+e^{\a r/7} s_\circ
    $$
    where $s=s(Z,\a/3)$ and $s_\circ=s_\circ(Z,\a/7)$.
\end{lemma}

\begin{proof}
    If $\s(x,y)<r$, then by the definition of $\s$ there exists $v \in Z$ such that 
    \begin{equation}\label{eq: xvy}
        \rho_1(x,v)<r \quad \text{and} \quad \rho_2(v,y) \le s.
    \end{equation}
    Using the first inequality in \eqref{eq: xvy} and then the definition of $s_\circ$, we get
    $$
    \rho_2(x,v) \le \diam_2\left(B_1(x,r)\right) \le e^{\a r/7} s_\circ.
    $$
    To complete the proof, combine this with the second inequality in \eqref{eq: xvy} by triangle inequality.
\end{proof}

\medskip

\begin{proof}[Proof of Theorem~\ref{thm: ultrametric full}]
    Combining Lemmas~\ref{lem: unrelaxation} and \ref{lem: relaxed accuracy}, we get for any $x \in Z$ and $r>0$:
    $$
    \Pr{\rho_2(x,M(x)) > s+e^{\a r/7} s_\circ}
    \le \Pr{\s(x,M(x)) \ge r}
    \lesssim e^{-\a r/6}.
    $$
    From this one can easily conclude that 
    $$
    \E \rho_2(x,M(x)) \lesssim s+s_\circ.
    $$
    The theorem is proved.
\end{proof}

\section{Examples} \label{s: examples}

The result below is quite standard. It can be derived from the $K$-norm mechanism in the paper \cite{hardttalwar}.
\begin{theorem}[Unit ball of a normed space]\label{thm: unit ball}
    Let $Z$ be the closed unit ball of a $d$-dimensional normed space $X$ with the induced metric. Then for any $\a>0$ we have 
    $$
    A(Z,\a) \asymp s(Z,\a) \asymp \min \left( \frac{d}{\a}, 1\right).
    $$
\end{theorem}

\begin{proof}
    According to Theorem~\ref{thm: same metric}, it is enough to prove the bound on the entropic scale $s(Z,\a)$.  
    Furthermore, by Proposition~\ref{prop: doubling}, it is enough to prove the same bound for the doubling scale $\ubar{s}(Z,\a)$.
    
    If $B_X(x,r)$ and $B_Z(x,r)$ denote the balls of $X$ and $Z$ centered at $x$ and with radius $r$, then by the definition of $Z$ we have
    $$
    B_Z(x,r) = B_X(x,r) \cap B_X(0,1).
    $$

    The rest of the argument is based on  standard entropy bounds (e.g.\ see~\cite[Chapter 5]{wainwright2019high}), namely
    $$
    2^d \le N \left(B_X(x,2),1\right) \le 5^d
    \quad \text{for any } x \in X
    $$
    
    To  prove the upper bound on doubling scale, note that
    $$
    N\left(B_Z(x,2s),s\right) 
    \le N\left(B_X(x,2s),s\right)
    = N\left(B_X(0,2),1\right)
    \le 5^d 
    \le e^{2\a s}
    $$
    for any $s \ge d/\a$, and 
    $$
    N\left(B_Z(x,2s),s\right) 
    \le N\left(B_X(0,1),s\right)
    = 1
    \le e^{2\a s}
    $$
    for any $s \ge 2$. Combining the two bounds yields
    $$
    \ubar{s}(Z,\a) \le \min\left(\frac{d}{\a}, 2\right).
    $$

    To prove the lower bound on the doubling scale, note that for any $s \le \frac{1}{10} \min\left(\frac{d}{\a}, 1\right)$ we have
    $$
    N\left(B_Z(0,2s),s\right)
    = N\left(B_X(0,2s),s\right)
    = N\left(B_X(0,2),1\right)
    \ge 2^d > e^{2\a s}.
    $$
    This yields 
    $$
    \ubar{s}(Z,\a) > \frac{1}{10} \min\left(\frac{d}{\a}, 1\right).
    $$
    The proof is complete. 
\end{proof}



\begin{theorem}[Probability measures with Wasserstein metric]\label{thm: Wasserstein}
    Consider the unit cube $[0,1]^d$ equipped with the $\norm{\cdot}_\infty$ metric. 
    Let $Z$ be the set of all probability measures on this cube, equipped with the $1$-Wasserstein metric. Then for any $\a>0$ we have
    $$
    A(Z,\a) \asymp s(Z,\a) \asymp \min\left( \a^{-\frac{1}{d+1}},1\right).
    $$
\end{theorem}

\begin{proof}
    Observe that $Z$ is a norm-convex space, since $Z$ is a convex subset of the normed space of all signed measures $\mu$ on $[0,1]^{d}$ with $\|\mu\|<\infty$, where
    \[\|\mu\|=\sup\left\{\left|\int f\,d\mu\right|:\;f\!:[0,1]^{d}\to\mathbb{R}\text{ is 1-Lipschitz}\right\},\]
    and the 1-Wasserstein metric between two probability measures $\mu$ and $\nu$ on $[0,1]^{d}$ coincides with $\|\mu-\nu\|$.
    
    According to Theorem~\ref{thm: same metric}, it is enough to prove the bound on the entropic scale $s(Z,\a)$. 
    Recall that 
    \begin{equation}\label{eq: measure packing}
        \exp \left( (c_1/\gamma)^d \right)
        \le N(Z,\gamma) 
        \le \exp \left( (c_2/\gamma)^d \right)
        \quad \text{for any }
        \gamma \in (0,1/3)
    \end{equation}
    where $c_1$ and $c_2$ are positive absolute constants. 
    For the lower bound in~\eqref{eq: measure packing} see
    Proposition 8.2 in~\cite{boedihardjo2022measure}.
   The upper bound follows from Lemma~\ref{lem: wcover} and the  equivalence between packing and covering numbers.
    Then one can easily conclude the desired bound for the outer scale, namely
    $$
    \bar{s}(Z,\a) \asymp \min\left( \a^{-\frac{1}{d+1}},1\right).    
    $$
    To transfer this result to the the entropic scale $s(Z,\a)$, we can use Proposition \ref{prop: sbar<s}. 
    Since the doubling condition $N(Z,\gamma) \ge N(Z,(2c_2/c_1)\gamma)^2$ holds for all $\gamma>0$, we have 
    $$
    s(Z,\a) \asymp \bar{s}(Z,\a).
    $$
\end{proof}

\begin{theorem}[Lipschitz functions] \label{thm: Lipschitz}
    Let $Z$ be the set of al $1$-Lipschitz functions $f:[0,1]^d \to \R$ satisfying $f(0)=0$, equipped with the $\norm{\cdot}_\infty$ metric. Then for any $\a>0$ we have
    $$
    A(Z,\a) \asymp s(Z,\a) \asymp \min\left( \a^{-\frac{1}{d+1}},1\right).
    $$
\end{theorem}

\begin{proof}
    The proof is nearly identical to that of Theorem~\ref{thm: Wasserstein}, since the two-sided bound \eqref{eq: measure packing} holds in this setting as well, see e.g.~\cite[Example~5.10]{wainwright2019high}.
\end{proof}



Next, we give an example for Theorem~\ref{thm: ultrametric}. 
Suppose that $f:\mathbb{N}\to(0,\infty)$ and $g:\mathbb{N}\to(0,\infty)$ are strictly decreasing with $\displaystyle\lim_{k\to\infty}f(k)=\lim_{k\to\infty}g(k)=0$. Define ultrametrics $\rho_{1},\rho_{2}$ on $\{0,1\}^{\mathbb{N}}$ by
\begin{equation}\label{ultra_rho1}
    \rho_1(x,y)
    =f \left( \inf \left\{j\in\mathbb{N} \;|\; x_j\neq y_j\right\}\right),
\end{equation}
and
\begin{equation}\label{ultra_rho2}
\rho_2(x,y)
    =g \left( \inf \left\{j\in\mathbb{N} \;|\; x_j\neq y_j\right\}\right).
\end{equation}

\begin{theorem}[Ultrametric]\label{thm:ultraexample}
Let $Z=(\{0,1\}^{\mathbb{N}},\rho_{1},\rho_{2})$ with $\rho_1, \rho_2$ as defined in~\eqref{ultra_rho1} and~\eqref{ultra_rho2}.
Then for any $\a>0$ we have
    \begin{equation}\label{ultraest}
      s(Z,\a) = g\left(\Big\lfloor\inf_{k\in\mathbb{N}}\big(k+\frac{\alpha}{\ln 2}f(k)\big)\Big\rfloor\right), \quad
      s_0(Z,\a) = 
    \sup_{k\in\mathbb{N}}g(k)e^{-\alpha f(k)}.
    \end{equation}
    
\end{theorem}

\begin{proof}

 Since the metric $\rho_2$ takes values in $\{g(k):\,k\in\mathbb{N}\}$ and the metric $\rho_{1}$ takes values in $\{f(k):\,k\in\mathbb{N}\}$,
\begin{eqnarray*}
s(Z,\alpha)&=&\inf\left\{s>0 \;\Big\vert\; N_2 \left(B_1(x,r),s\right)\leq e^{\alpha r} \; \forall x\in\{0,1\}^{\mathbb{N}}, \; \forall r>0 \right\}\\&=&
\inf\left\{g(k_{0}) \;\Big\vert\; N_2 \left(B_1(x,f(k)),g(k_{0})\right)\leq e^{\alpha f(k)} \; \forall x\in\{0,1\}^{\mathbb{N}}, \; \forall k\in\mathbb{N} \right\}.
\end{eqnarray*}
It is easy to see that
\begin{equation}\label{eq: B1}
    B_1(x,f(k))=\left\{y\in\{0,1\}^{\mathbb{N}}:\,x_{j}=y_{j}\;\forall j\leq k-1\right\}.   
\end{equation}
Also, for $x,y\in\{0,1\}^{\mathbb{N}}$, we have $\rho_{2}(x,y)>g(k_{0})$ if and only if $\inf\{j\in\mathbb{N}:\,x_{j}\neq y_{j}\}<k_{0}$.
So
\[N_2 \left(B_1(x,f(k)),g(k_{0})\right)=\begin{cases}2^{k_{0}-k},&k<k_{0}\\1,&k\geq k_{0}\end{cases},\]
for all $x\in\{0,1\}^{\mathbb{N}}$. Therefore,
\begin{eqnarray*}
s(Z,\alpha)&=&\inf\left\{g(k_{0}) \;\Big\vert\; 2^{k_{0}-k}\leq e^{\alpha f(k)} \; \forall k<k_{0} \right\}\\&=&
\inf\left\{g(k_{0}) \;\Big\vert\; k_{0}\leq\inf_{k\in\mathbb{N}}\Big(k+\frac{\alpha}{\ln 2}f(k)\Big)\right\}\\&=&
g\left(\Big\lfloor\inf_{k\in\mathbb{N}}\big(k+\frac{\alpha}{\ln 2}f(k)\big)\Big\rfloor\right).
\end{eqnarray*}
Now we consider $s_{\circ}(Z,\alpha)$. We have
\begin{eqnarray*}
s_{\circ}(Z,\alpha) &=& \inf\left\{s>0 \;\Big\vert\;
    \diam_2\left(B_1(x,r)\right)\leq e^{\alpha r} s \; \forall x\in Z, \; \forall r>0 \right\}\\
    &=&
    \inf\left\{s>0 \;\Big\vert\;  \diam_2\left(B_1(x,f(k))\right)\leq e^{\alpha f(k)} s \; \forall x\in Z, \; \forall k\in\mathbb{N} \right\}.
\end{eqnarray*}
By \eqref{eq: B1}, we have 
$\mathrm{diam}_2 \left(B_1(x,f(k))\right)=g(k)$. So
\begin{eqnarray*}
s_{\circ}(Z,\alpha) &=& \inf\left\{s>0 \;\Big\vert\;
    g(k)\leq e^{\alpha f(k)} s \; \forall k\in\mathbb{N} \right\}\\
    &=&
    \sup_{k\in\mathbb{N}}g(k)e^{-\alpha f(k)}. \qedhere
\end{eqnarray*}

\end{proof}

\begin{remark}[Baire metric]\label{rem: Baire}
The ultrametrics defined in~\eqref{ultra_rho1} and~\eqref{ultra_rho2} 
include as special case the {\em Baire metric}~\cite{contreras2012fast,dragovich2017ultrametrics}, which is obtained by letting
 $f(k) = r^{-k}$ and $g(k) = r^{-k}$ for some (not necessarily the same) $r>1$. 
If we further set $r=2$ for both metrics, we obtain from equation~\eqref{ultraest} after some brief computations that
$$
s(Z,\a) \asymp \min\left(\frac{1}{\alpha},1\right)
\quad \text{and} \quad
s_0(Z,\a) \asymp\min\left(\frac{1}{\alpha},1\right).
$$
Due to Theorem~\ref{thm: ultrametric}, this establishes~\eqref{ultraboolean}.
\end{remark}

\appendix
\section{The metric entropy of the set of all probability measures}

\begin{lemma}\label{lem: wcover}
Let $0<\gamma<1$. The $\gamma$-covering number of the set $Z$ of all probability measures on $[0,1]^{d}$ with respect to the $1$-Wasserstein distance is at most $\exp((\frac{5}{\gamma})^{d})$.
\end{lemma}
\begin{proof}
Let $n=\lceil\frac{2}{\gamma}\rceil$ and $S=\{\frac{1}{n},\ldots,\frac{n}{n}\}^{d}$. We denote
\[\Lambda=\left\{\sum_{x\in S}a_{x}\delta_{x}:\,\sum_{x\in S}a_{x}=1,\;a_{x}\in\left\{\frac{0}{|S|},\frac{1}{|S|},\ldots,\frac{|S|}{|S|}\right\}\;\forall x\in S\right\}.\]
Then $|\Lambda|$ coincides with the number of ways to distribute $|S|$ unlabeled balls into $|S|$ bins. So
\[|\Lambda|=\binom{2|S|-1}{|S|-1}\leq 4^{|S|}=\exp(n^{d}\ln 4)\leq\exp\left(\left(\frac{5}{\gamma}\right)^{d}\right).\]
It remains to show that every probability measure $\mu$ on $[0,1]^{d}$ has 1-Wasserstein distance at most $\frac{2}{n}$ from a measure in $\Lambda$. First in view of the set $S$, there is a probability measure $\nu$ supported on $S$ such that $W_{1}(\mu,\nu)\leq\frac{1}{n}$. We now show that $\nu$ has 1-Wasserstein distance at most $\frac{1}{n}$ from a measure in $\Lambda$.

By induction on $d$, one can construct an enumeration $x_{1},\ldots,x_{|S|}$ of $S$ such that
\[\sum_{k=1}^{|S|-1}\|x_{k+1}-x_{k}\|_{\infty}\leq n^{d-1}-\frac{1}{n}\leq n^{d-1}.\]
In other words, a ``travelling salesman" can visit every point in $S$ by travelling a total distance of at most $n^{d-1}$.

We now construct a measure in $\Lambda$ based on the measure $\nu$ as follows. Write $\nu(\{x_{1}\})=\frac{m_{1}+\omega_{1}}{|S|}$ where $m_{1}\in\{0,\ldots,|S|-1\}$ and $0\leq\omega_{1}\leq 1$. The salesman moves the weight $\frac{\omega_{1}}{|S|}$ from $x_{1}$ to $x_{2}$.

Next write $\nu(\{x_{2}\})+\frac{\omega_{1}}{|S|}=\frac{m_{2}+\omega_{2}}{|S|}$ where $m_{2}\in\{0,\ldots,|S|-1\}$ and $0\leq\omega_{2}\leq 1$. The salesman moves the weight $\frac{\omega_{2}}{|S|}$ from $x_{2}$ to $x_{3}$. Continue until the last step: the salesman moves the weight $\frac{\omega_{|S|-1}}{|S|}$ from $x_{|S|-1}$ to $x_{|S|}$. Because all the probabilities sum up to $1$, the new weight of $x_{|S|}$ must be an integer multiple of $\frac{1}{|S|}$. So after moving all the weights, the new probability measure is in $\Lambda$.

Finally let us sum up the weights times the distances moved. We have
\[\sum_{k=1}^{|S|-1}\frac{\omega_{k}}{|S|}\|x_{k+1}-x_{k}\|_{\infty}\leq\sum_{k=1}^{|S|-1}\frac{1}{|S|}\|x_{k+1}-x_{k}\|_{\infty}\leq\frac{n^{d-1}}{|S|}=\frac{1}{n}. \qedhere\] 
\end{proof}


\begin{thebibliography}{10}

\bibitem{ABCP}
Miguel~E Andr{\'e}s, Nicol{\'a}s~E Bordenabe, Konstantinos Chatzikokolakis, and
  Catuscia Palamidessi.
\newblock Geo-indistinguishability: Differential privacy for location-based
  systems.
\newblock In {\em Proceedings of the 2013 ACM SIGSAC conference on Computer \&
  communications security}, pages 901--914, 2013.

\bibitem{bellovin2019privacy}
Steven~M Bellovin, Preetam~K Dutta, and Nathan Reitinger.
\newblock Privacy and synthetic datasets.
\newblock {\em Stan. Tech. L. Rev.}, 22:1, 2019.

\bibitem{boedihardjo2022covariance}
March Boedihardjo, Thomas Strohmer, and Roman Vershynin.
\newblock Covariance’s loss is privacy’s gain: Computationally efficient,
  private and accurate synthetic data.
\newblock {\em Foundations of Computational Mathematics}, 24:179--226, 2024.

\bibitem{boedihardjo2022measure}
March Boedihardjo, Thomas Strohmer, and Roman Vershynin.
\newblock Private measures, random walks, and synthetic data.
\newblock {\em Probability Theory and Related Fields}, to appear.

\bibitem{bordenabe2014optimal}
Nicol{\'a}s~E Bordenabe, Konstantinos Chatzikokolakis, and Catuscia
  Palamidessi.
\newblock Optimal geo-indistinguishable mechanisms for location privacy.
\newblock In {\em Proceedings of the 2014 ACM SIGSAC conference on computer and
  communications security}, pages 251--262, 2014.

\bibitem{chatzikokolakis2013broadening}
Konstantinos Chatzikokolakis, Miguel~E Andr{\'e}s, Nicol{\'a}s~Emilio
  Bordenabe, and Catuscia Palamidessi.
\newblock Broadening the scope of differential privacy using metrics.
\newblock In {\em Privacy Enhancing Technologies: 13th International Symposium,
  PETS 2013, Bloomington, IN, USA, July 10-12, 2013. Proceedings 13}, pages
  82--102. Springer, 2013.

\bibitem{contreras2012fast}
Pedro Contreras and Fionn Murtagh.
\newblock Fast, linear time hierarchical clustering using the baire metric.
\newblock {\em Journal of Classification}, 29:118--143, 2012.

\bibitem{dragovich2017ultrametrics}
Branko Dragovich, Andrei~Yu Khrennikov, and Nata{\v{s}}a~{\v{Z}}
  Mi{\v{s}}i{\'c}.
\newblock Ultrametrics in the genetic code and the genome.
\newblock {\em Applied Mathematics and Computation}, 309:350--358, 2017.

\bibitem{dwork2014algorithmic}
Cynthia Dwork and Aaron Roth.
\newblock The algorithmic foundations of differential privacy.
\newblock {\em Foundations and Trends in Theoretical Computer Science},
  9(3-4):211--407, 2014.

\bibitem{fernandes2019generalised}
Natasha Fernandes, Mark Dras, and Annabelle McIver.
\newblock Generalised differential privacy for text document processing.
\newblock In {\em Principles of Security and Trust: 8th International
  Conference, POST 2019, Held as Part of the European Joint Conferences on
  Theory and Practice of Software, ETAPS 2019, Prague, Czech Republic, April
  6--11, 2019, Proceedings 8}, pages 123--148. Springer International
  Publishing, 2019.

\bibitem{fernandes2022universal}
Natasha Fernandes, Annabelle McIver, Catuscia Palamidessi, and Ming Ding.
\newblock Universal optimality and robust utility bounds for metric
  differential privacy.
\newblock In {\em 2022 IEEE 35th Computer Security Foundations Symposium
  (CSF)}, pages 348--363. IEEE, 2022.

\bibitem{galli2023advancing}
Filippo Galli, Kangsoo Jung, Sayan Biswas, Catuscia Palamidessi, and Tommaso
  Cucinotta.
\newblock Advancing personalized federated learning: Group privacy, fairness,
  and beyond.
\newblock {\em SN Computer Science}, 4(6):831, 2023.

\bibitem{geng2014optimal}
Quan Geng and Pramod Viswanath.
\newblock The optimal mechanism in differential privacy.
\newblock In {\em 2014 IEEE international symposium on information theory},
  pages 2371--2375. IEEE, 2014.

\bibitem{ghosh2009universally}
Arpita Ghosh, Tim Roughgarden, and Mukund Sundararajan.
\newblock Universally utility-maximizing privacy mechanisms.
\newblock In {\em Proceedings of the forty-first annual ACM symposium on Theory
  of computing}, pages 351--360, 2009.

\bibitem{gupte2010universally}
Mangesh Gupte and Mukund Sundararajan.
\newblock Universally optimal privacy mechanisms for minimax agents.
\newblock In {\em Proceedings of the twenty-ninth ACM SIGMOD-SIGACT-SIGART
  symposium on Principles of database systems}, pages 135--146, 2010.

\bibitem{hardttalwar}
Moritz Hardt and Kunal Talwar.
\newblock On the geometry of differential privacy.
\newblock In {\em Proceedings of the 42nd ACM symposium on Theory of computing,
  STOC '10}, pages 705--714, New York, NY, USA, 2010.

\bibitem{imola2022balancing}
Jacob Imola, Shiva Kasiviswanathan, Stephen White, Abhinav Aggarwal, and
  Nathanael Teissier.
\newblock Balancing utility and scalability in metric differential privacy.
\newblock In {\em Uncertainty in Artificial Intelligence}, pages 885--894.
  PMLR, 2022.

\bibitem{nikolov2023private}
Aleksandar Nikolov.
\newblock Private query release via the johnson-lindenstrauss transform.
\newblock In {\em Proceedings of the 2023 Annual ACM-SIAM Symposium on Discrete
  Algorithms (SODA)}, pages 4982--5002. SIAM, 2023.

\bibitem{nikolov2013geometry}
Aleksandar Nikolov, Kunal Talwar, and Li~Zhang.
\newblock The geometry of differential privacy: the sparse and approximate
  cases.
\newblock In {\em Proceedings of the forty-fifth annual ACM symposium on Theory
  of computing}, pages 351--360, 2013.

\bibitem{ullman2011pcps}
Jonathan Ullman and Salil Vadhan.
\newblock {PCPs} and the hardness of generating private synthetic data.
\newblock In {\em Theory of Cryptography Conference}, pages 400--416. Springer,
  2011.

\bibitem{vershyninbook}
Roman Vershynin.
\newblock {\em High-dimensional probability. An introduction with applications
  in data science}.
\newblock Cambridge University Press, 2018.

\bibitem{wainwright2019high}
Martin~J Wainwright.
\newblock {\em High-dimensional statistics: A non-asymptotic viewpoint},
  volume~48.
\newblock Cambridge university press, 2019.

\end{thebibliography}

\end{document}